\documentclass{article}
\usepackage[utf8]{inputenc}
\usepackage[T1]{fontenc}
\usepackage[english]{babel}

\usepackage{blindtext}
\usepackage{amssymb,amsmath,amsthm}
\usepackage[linesnumbered,algoruled,boxed,lined]{algorithm2e}
\usepackage{enumerate} 
\usepackage{a4wide}
\usepackage{yfonts}
\allowdisplaybreaks
\usepackage[active]{srcltx}
\usepackage{comment}
\usepackage{epsfig}
\usepackage{xspace}
\usepackage{fancyhdr}
\usepackage{url}
\usepackage{hyperref}
\usepackage{color}
\usepackage[pass]{geometry}
\usepackage[normalem]{ulem}
\usepackage{authblk}

\usepackage[textsize=small]{todonotes}

\newtheorem{theorem}{Theorem}[section]
\newtheorem{corollary}[theorem]{Corollary}
\newtheorem{lemma}[theorem]{Lemma}
\newtheorem{proposition}[theorem]{Proposition}
\newtheorem{claim}[theorem]{Claim}

\theoremstyle{remark}
\newtheorem{remark}[theorem]{Remark}

\theoremstyle{definition}
\newtheorem{definition}[theorem]{Definition}

\providecommand{\keywords}[1]
{
  \small	
  \textbf{\textit{Keywords---}} #1
}

\newcommand{\ham}{\textsc{Hamiltonian Path}\xspace}
\newcommand{\domiset}{\textsc{Dominating Set}\xspace}
\newcommand{\delcactus}{\textsc{Deletion to cacti}\xspace}
\newcommand{\delcater}{\textsc{Deletion to caterpillar}\xspace}
\newcommand{\delcons}{\textsc{Deletion to constellations}\xspace}

\title{Edge deletion to tree-like graph classes}

\newcommand{\suggestion}[2]{\!#2\!}

\author[1]{Ivo Koch}
\author[2,3,5]{Nina Pardal}
\author[4]{Vinicius Fernandes dos Santos}
\affil[1]{University of General Sarmiento, Argentina}
\affil[2]{University of Buenos Aires, Argentina}
\affil[3]{ICC-CONICET, Argentina}
\affil[4]{Computer Science Department, Federal University of Minas Gerais, Brazil}
\affil[5]{University of Sheffield, UK}

\date{}

\begin{document}

\maketitle


\begin{abstract}
For a fixed property (graph class) $\Pi$, given a graph $G$ and an integer $k$, the \emph{$\Pi$-deletion} problem consists in deciding if we can turn $G$ into a graph with the property $\Pi$ by deleting at most $k$ edges. 
The $\Pi$-deletion problem is known to be NP-hard for most of the well-studied graph classes, such as chordal, interval, bipartite, planar, comparability and permutation graphs, among others; even deletion to \emph{cacti} is known to be NP-hard for general graphs. However, there is a notable exception: the deletion problem to \emph{trees} is polynomial.
Motivated by this fact, we study the deletion problem for some classes \suggestion{close}{similar} to trees, addressing in this way a knowledge gap in the literature. We prove that deletion to cacti is hard even when the input is a bipartite graph. On the positive side, we show that the problem becomes tractable when the input is chordal, and for the special case of quasi-threshold graphs we give a simpler and faster algorithm.
In addition, we present sufficient structural conditions on the graph class $\Pi$ that imply the NP-hardness of the $\Pi$-deletion problem, and show that deletion from general graphs to some well-known subclasses of \suggestion{trees}{forests} is NP-hard.
\end{abstract}



\section{Introduction}

A graph modification problem consists in deleting or adding edges or vertices to a graph so that the resulting graph fulfills some pre-specified properties. More formally, for a fixed property $\Pi$ that usually represents membership to a graph class, given a graph $G$ and an integer $k$, the $\Pi$-deletion (resp.\ completion, editing) problem consists in deciding whether it is possible to modify the graph by deleting (resp.\ adding, or adding and deleting) at most $k$ edges to obtain a new graph that fullfils the property $\Pi$. 
There is one more version of this type of problems, which is the $\Pi$-vertex-deletion, that consists in determining if there exists an induced subgraph with the property $\Pi$, obtained by deleting at most $k$ vertices.
These are the decision versions of graph modification problems, which also admit an optimization version. When $k$ is not a part of the input, the most natural goal in graph modification problems is to find a \emph{minimum} modification, that is, to determine the smallest subset of edges or vertices that one has to delete or add to obtain a graph with the desired property.

Graphs are very useful to model quite diverse real world and theoretical structures. In particular, graph modification can be used to model and solve a wide variety of problems in many dissimilar fields, such as database and inconsistency management~\cite{TarYan1984}, computer vision~\cite{ChuMum1994}, and molecular biology~\cite{GolGolKapSha1995,GolKapSha1994}, to name a few.
Furthermore, many fundamental problems in graph theory can be reinterpreted as graph modification problems. For example, the Connectivity problem can be thought of as the problem of finding the minimum number of vertices or edges that disconnect the graph when removed from it.

Most graph modification problems turn out to be intractable.
Yannakakis~\cite{Yannakakis79} proved that finding a maximum connected \suggestion{}{induced} subgraph with a property $\Pi$ is NP-hard not only for every ``non-trivial'' hereditary property $\Pi$, but also for other properties that are not hereditary. One example of this are trees and stars: even though trees and stars are not hereditary classes, every connected induced subgraph of a tree is a tree and every connected induced subgraph of a star is a star. As a consequence of the aforementioned result by Yannakakis, removing the minimum number of vertices to obtain a tree or a star turns out to be NP-hard.

As for the $\Pi$-deletion problem, there are no such general results. 
The relationship between the vertex and the edge version of the problem is also not consistent. One example we might consider is the property $\Pi$ of being a forest. In this case, the $\Pi$-deletion problem is equivalent to that of finding a spanning tree for each connected component of the input graph, which is easy. Whereas it follows from the results by Yannakakis that minimum vertex deletion for an induced forest is NP-hard.
Furthermore, concerning the $\Pi$-deletion problem there are some NP-hardness results that hold for some quite large families of hereditary properties, such as the property of being cluster, $P_k$-free,\cite{Colburn-ElMallah88}, \suggestion{}{or $C_k$-free~\cite{yannakakis_edge_deletion} for some fixed $k\geq 3$.}
An interesting characterization of the $\Pi$-deletion problem for a large and natural family of graph properties can be found in~\cite{AlonSS05}, similar to the one given by Yannakakis. In that paper, the authors show the NP-hardness of the $\Pi$-deletion problem for every monotone property $\Pi$ \suggestion{}{(i.e. a graph property closed under removal of vertices and edges)} that holds for all bipartite graphs. This encompasses many interesting properties, for example the one of being triangle-free.

The $\Pi$-deletion problem turns out to be hard for most graph classes (e.g.\ planar, chordal, bipartite, interval and proper interval, split, threshold, etc), but is polynomially solvable for trees. This work is motivated by the exploration of the tractable-intractable threshold of the $\Pi$-deletion problem when $\Pi$ is a graph class closely related to trees, namely constellations, caterpillars, and cacti. 

We find that the decision version of all these problems is NP-complete, even \suggestion{for bipartite graphs}{when the input is a bipartite graph.} \suggestion{}{Though the general problem for cacti is already known to be hard~\cite{Colburn-ElMallah88}, this refines the intractability result for this graph class, while providing new results for the deletion problem into the other aforementioned tree-like graph classes}. We also provide sufficient structural properties for the graph class $\Pi$ to ensure that $\Pi$-deletion is NP-hard, even restricted to some subclasses of bipartite graphs. \suggestion{}{Even though not as general as those of Yannakakis~\cite{Yannakakis79}, this provides a framework for showing hardness results for some non-trivial graph classes.} Then we delve into cactus deletion and prove that this problem can be solved in polynomial time when the input is restricted to chordal graphs, and give a simple and efficient polynomial-time algorithm for quasi-threshold graphs. 

This article is organized as follows. Section~\ref{sec:defi} presents basic definitions, in Section~\ref{sec:edge-del-complexity} start by proving NP-completeness for cactus deletion when the input is restricted to bipartite graphs. Then, in Section~\ref{sec:generalization}, we propose sufficient structural properties for a graph class $\Pi$ so that the $\Pi$-deletion problem results NP-hard even within bipartite graphs. We finish the section by showing NP-completeness for constellations and caterpillars, two known subclasses of forests.  Section~\ref{sec:algorithms} focuses on positive results: we start by showing that cactus deletion is polynomial-time solvable when the input graph $G$ is chordal; a specialized, efficient procedure is presented for quasi-threshold graphs in Section~\ref{sec:qt_cactus}. 
Finally, the last Section~\ref{sec:conclusions} is devoted to conclusions and future work directions.

\section{Definitions}\label{sec:defi}

All the graphs in this paper are undirected and simple, unless explicitly stated otherwise. Let $G$ be a graph, and let $V(G)$ and $E(G)$ denote its vertex and edge sets, respectively. We denote by $n$ the number of vertices and by $m$ the number of edges. Further, we use the standard notation $N(V')$ for the neighborhood of a subset $V'$ of $V$.
Whenever it is clear from the context, we simply write $V$ and $E$ and note $G=(V,E)$. For basic definitions not included here, we refer the reader to~\cite{Bondy}. 

Given a graph $G$ and $S\subseteq V$, the \emph{subgraph of $G$ induced by $S$}, denoted by $G[S]$, is the graph with vertex set $S$ such that two vertices of $S$ are adjacent if and only if they are adjacent in $G$. When $G'$ and $G[S]$ are isomorphic for some $S \subseteq V$, by abuse of terminology we say that $G'$ is an \emph{induced subgraph} of $G$. If $V' \subseteq V$ and $E' \subseteq E$, we simply say that $H=(V', E')$ is a \emph{subgraph of $G$}, whether $E'$ consists of all the edges between vertices in $V'$ or only a subset of those.
For any family $\mathcal{F}$ of graphs, we say that $G$ is \emph{$\mathcal{F}$-free} if $G$ does not contain any graph $F\in \mathcal{F}$ as an induced subgraph. If $F\in \mathcal{F}$ is the unique element of $ \mathcal{F}$, we may write \emph{${F}$-free}, instead.
If a graph $G$ is $\mathcal{F}$-free, then the graphs in $\mathcal{F}$ are called the \emph{forbidden induced subgraphs} of $G$.
A graph property $\Pi$ is \emph{hereditary} if $\Pi$ is inherited by every induced subgraph of a graph.

A graph is \emph{complete} if all its vertices are pairwise adjacent. A vertex is \emph{universal} if it is adjacent to all the vertices in a graph.
A \emph{clique} in a graph $G$ is a complete induced subgraph of $G$. Also, we will often use this term for the vertex set that induces the clique. We denote by $K_n$ the complete graph of $n$ vertices.

Let $G_1 = (V_1, E_1)$ and $G_2 = (V_2,E_2)$ be two graphs with $V_1 \cap V_2 = \emptyset$. The \emph{union} of $G_1$ and $G_2$ is the graph $G_1 \cup G_2 = (V_1 \cup V_2, E_1 \cup E_2)$, and the \emph{join} of $G_1$ and $G_2$ is the graph $G_1 \oplus G_2 = (V_1 \cup V_2 , E_1 \cup E_2 \cup V_1 \times V_2 )$.
A graph $G$ is a \emph{cograph} if it is $P_4$-free. Any cograph can be constructed from $K_1$ by repeated application of union and join operations. A graph is \emph{quasi-threshold} if it is $\{P_4, C_4\}$-free. These graphs are also known as \emph{trivially perfect graphs}. Analogously as for cographs, quasi-threshold graphs can be defined recursively from $K_1$ by either repeatedly considering disjoint unions, or adding a universal vertex.

A graph is \emph{cluster} if it is the union of cliques. A graph is a \emph{block graph} if every maximal $2$-connected component (i.e. every \textit{block}) is a clique. 

\suggestion{}{A graph is \emph{subcubic} if every vertex has degree at most three.}

A graph $G$ is an \emph{interval} graph if it admits an intersection model consisting of intervals on the real line, that is, a family ${\cal I}$ of intervals on the real line and a one-to-one correspondence between the set of vertices of $G$ and the intervals of ${\cal I}$ such that two vertices are adjacent in $G$ if and only if the corresponding intervals intersect. 
A \emph{proper interval} graph is an interval graph that admits a \emph{proper interval model}, this is, an intersection model in which no interval is properly contained in any other. 

A graph $G$ is a \emph{bipartite} graph if its vertex set can be partitioned into two sets $V_1, V_2$ of pairwise nonadjacent vertices. A bipartite graph is \emph{complete} if every pair of vertices $v_1 \in V_1, v_2 \in V_2$ is an edge of the graph and, if $|V_1|=n$ and $|V_2|=m$, we denote it by $K_{n,m}$.

A \emph{forest} is an acyclic graph. A \emph{tree} is a connected acyclic graph. The vertices of degree 1 of a graph are called \emph{pendant vertices} and pendant vertices of a tree are called \emph{leaves}. A \emph{path} is a sequence of vertices $v_1, v_2, \ldots, v_k$ such that $v_i$ and $v_{i+1}$ are adjacent for each $1\leq i\leq k-1$ and $v_i v_j$ is not an edge for every pair of non-consecutive vertices $v_i$ and $v_j$. We denote by $P_k$ a path with $k$ vertices.

A graph $G$ is a \emph{caterpillar} if $G$ is a tree in which the removal of all the leaves results in a path (called the \emph{spine} or \emph{central path}).

A \emph{star} is a complete bipartite graph $K_{1,n}$ for some $n\geq 0$. A \emph{constellation} is a forest where each connected component is a star.

A graph $G$ is a \emph{cactus} if $G$ is connected and every edge of the graph lies in at most one cycle.
\suggestion{}{If each connected component of a graph $G$ induces a cactus, we say that $G$ is a \emph{forest of cacti}.}
Given a graph class $\Pi$, we say that a subgraph $H$ of $G$ is a \emph{maximum spanning $\Pi$} if $H$ belongs to $\Pi$ and the number of edges is maximum among all the $\Pi$-subgraphs of $G$.
A graph $G$ is \emph{chordal} if it contains no induced cycles of length greater than $3$.

A graph $G$ is \emph{planar} if it admits a planar embedding, this is, there is a way to draw it on the plane such that its edges intersect only at their endpoints. 

A \emph{hamiltonian path (resp.\ cycle)} is a path (resp.\ cycle) that visits each vertex of a graph exactly once.
A \emph{dominating set} of a graph $G=(V,E)$ is a subset $D \subseteq V$ such that each vertex in $V \setminus D$ is adjacent to at least one vertex in $D$.

Given a graph $G$, the decision problem \ham consists in deciding if there is a hamiltonian path in $G$, or not.
For its part, given a graph $G$ and a integer $k$, the decision problem \domiset consists in determining whether there is a dominating set of $G$ of size at most $k$.
Both problems are known to be NP-complete, even when the input is restricted to subclasses of bipartite or chordal graphs (see, for example, \cite{bertossi1984dominating,bertossi1986hamiltonian,corneil1984clustering,itai1982hamilton,zvervich1995induced}).

\section{Complexity results for edge-deletion to classes close to trees}\label{sec:edge-del-complexity}

This section is organized as follows. First, in Section~\ref{sec:cactus_bipartite} we study the complexity of the edge-deletion problem to cacti when the input is a bipartite graph. In Section~\ref{sec:generalization} we generalize the hardness results for $\Pi$-deletion when the input is a bipartite graph by obtaining sufficient structural conditions for the target class $\Pi$. 
\suggestion{We finish by showing in Section~\ref{sec:subclasses_trees}
that the problem is intractable for constellations and caterpillars.}{Finally, Section~\ref{sec:subclasses_trees}
is devoted to showing that the problem is intractable when $\Pi$ is either the class of constellations or caterpillars.}

\subsection{Cacti}\label{sec:cactus_bipartite}

Edge-deletion to cacti is known to be NP-hard~\cite{Colburn-ElMallah88}. Our first target is to shed some light on how should we restrict the input to obtain tractability for cactus deletion.


\begin{theorem}
{\sc Deletion to cactus} is NP-complete for bipartite graphs.
\end{theorem}

\begin{proof}
Membership in NP is clear.
For the NP-hardness, we reduce from {\sc Partition into $P_3$}: 

Given a graph $G=(V,E)$, the problem {\sc Partition to $P_3$ (PIP3)} consists in deciding whether there exists a partition of $V$ into vertex-disjoint $P_3$'s. {\sc PIP3} is known to be NP-hard, even for subcubic bipartite graphs~\cite{monnot2007path}.


Let $G$ be a subcubic bipartite instance of {\sc PIP3}, with bipartition $\{X,Y\}$. We may assume that $|V(G)|$ is divisible by 3, otherwise the answer is trivially negative. Let $H$ be a graph with
\begin{itemize}
\item $V(H) = V(G) \cup \{x,y,x',y'\}$,
\item $E(H) = E(G) \cup \{xv \mid v \in X\} \cup \{yv \mid v \in Y\} \cup \{xy, xx', x'y', yy'\}$. 
\end{itemize}
Notice that $H$ is bipartite. Let $k = |E(H)| - 4|V(G)|/3 - 4$. 
We will show that there is a partition into $P_3$'s for $G$ if and only if there is a spanning cactus subgraph of $H$ that is obtained by removing at most $k$ edges from $E(H)$.

First, suppose $G$ has a partition $\mathcal{P}$ into (induced) $P_3$'s, and let $E(\mathcal{P})$ be the set of edges appearing in $\mathcal{P}$. Let $S$ be the set of endpoints of each $P_3$ in $\mathcal{P}$, $S_X = S \cap X$, and $S_Y = S \cap Y$. Now, let
$$
E' = E(\mathcal{P}) \cup \{xv \mid v \in S_X\} \cup \{yv \mid v \in S_Y\} \cup \{xy, xx', x'y', yy'\}.
$$
It is easy to see that $H' = (V(H), E')$ is a cactus with $4|V(G)|/3 + 4$ edges, where the cycles of $H'$ are the union of a $P_3$ in $G$ and $x$ or $y$. In particular, $H'$ is a subgraph of $H$, obtained by deleting exactly $k$ edges from $H$.

For the other direction, let $H'$ be a cactus subgraph of $H$ obtained by deleting at most $k$ edges. From the definition of $k$, it follows that $|E(H')| \geq 4|V(G)|/3 + 4$, thus implying $H'$ is not acyclic. Let $c$ the number of cycles in $H'$, and let $\ell_1 , \ldots , \ell_c$ be their respective lengths. 

\begin{claim}\label{claim:cactus_nph}
The number of cycles $c$ in $H'$ equals $|V(G)|/3 + 1 = (|V(H')|-1)/3$.
\end{claim}

\begin{proof}[\textit{Proof of Claim~\ref{claim:cactus_nph}}]
Deleting an edge of each cycle of $H'$ would turn it into a tree. Hence, $|E(H')| = |V(H')| + c - 1 = |V(G)| + c + 3$. 
On the other hand, $|E(H')| \geq 4|V(G)|/3 + 4$, thus it follows that $|V(G)| + c + 3 \geq 4|V(G)|/3 + 4$. Hence, this implies that $c \geq |V(G)|/3 + 1 = (|V(H')|-1)/3$.

Since $H'$ is bipartite, $\ell_i \geq 4$ for every cycle in $H'$, thus it follows that $|E(H')| \geq  \sum \ell_i \geq 4c$. But then $|V(H')| + c - 1 \geq 4c$, which implies $c \leq (|V(H')|-1)/3$.
This finishes the proof of the claim.
\end{proof}

Notice that by replacing $c$ in the equation $|E(H')| = |V(H')| + c - 1$, it follows that exactly $k$ edges were removed from $H$. Also note that the equality must hold in $|E(H')| \geq  \sum \ell_i \geq 4c$, implying that $\ell_1 = \ldots = \ell_c = 4$ and that each edge belongs to exactly one $C_4$. Since $x'$ and $y'$ belong to exactly one cycle, namely $\mathcal{C}= xx'y'yx$, all the other cycles must contain vertices from $G$.
Finally, every degree 2 vertex of each cycle $C$ in $H'$ that is not $\mathcal{C}$ lies in $G$. Consider a partition into $P_3$'s as follows. For each cycle $C$ such that either $x \in C$ or $y \in C$, take the $P_3$ consisting of $C \setminus \{x\}$ or $C \setminus \{y \}$, respectively. 
We will show that for each $C$ such that all its vertices lie in $G$, exactly one of its vertices $v$ lies in another cycle of the cactus. Notice that exactly one edge incident to $v$ is either $xv$ or $yv$. Suppose without loss of generality that this is $xv$.
Since $G$ is a subcubic graph, each vertex $v$ in such a cycle $C$ lies in at most two cycles of the cactus. More precisely, $v$ belongs to the cycle $C$ and at most to a cycle $C'$ given by traversing the remaining neighbour $w$ of $v$ in $G$, if there is one, and then completing the cycle traversing through $x$, $x'$, $y$, $y'$ and then back to $v$. 
Suppose that there are two vertices $v$ and $v'$ in $C$ such that both lie in other cycles $D$ and $D'$ like the ones described in the previous paragraph. Then, the edges of the $vv'$-path in $C$ also belong to another cycle, namely the one obtained by concatenating the paths $v\ldots v'$, $v'x$ and $vx$, using edges of $E(H') \setminus C$, a contradiction.
Hence, if $v$ is the only vertex of such a cycle $C$ belonging to another cycle, we may take the $P_3$ consisting of $C \setminus \{v\}$. This gives a partition into $P_3$'s of $G$.
\end{proof}

\subsection{Sufficient conditions for hardness results within bipartite graphs}\label{sec:generalization} 

The previous edge-deletion problem remains hard even if the input graph is bipartite. In this section, we give a set of sufficient conditions for a property $\Pi$ such that $\Pi$-deletion remains hard within bipartite graphs. 

First, notice these useful facts:

\begin{remark}\label{obs:conditions1}
Edge-deletion maintains bipartiteness.
\end{remark}

\begin{remark}\label{obs:conditions2}
A bipartite claw-free acyclic graph is the union of paths.
\end{remark}

The following result summarizes necessary conditions for the graph class $\Pi$ in order to obtain hardness for the $\Pi$-deletion problem.

\begin{theorem}\label{teo:structural_prop}
Let $\Pi$ be a graph class such that all of the following conditions hold:
\begin{enumerate}
    \item[(i)] $\Pi$ contains $P_n$, for all $n$;
    \item[(ii)] $\Pi$ is claw-free;
    \item[(iii)] $\Pi$ is even-hole-free.
\end{enumerate}

Then, $\Pi$-deletion is NP-hard, even when the input is restricted to planar bipartite graphs of maximum degree 3, or chordal bipartite graphs.
\end{theorem}

\begin{proof}
\suggestion{}{Assume that the input is a bipartite graph.} It follows from Remark~\ref{obs:conditions1}, $(ii)$ and $(iii)$ that the final graph obtained after $\Pi$-deletion will be a bipartite claw-free acyclic graph. Moreover, it follows from Remark~\ref{obs:conditions2} and $(i)$ that a solution exists and that the resulting graph is the union of paths.

For the NP-hardness, we give a reduction from \ham. 

Let us consider a bipartite graph $G$. From the previous observations, any subgraph belonging to $\Pi$ has at most $|V(G)|-1$ edges. Furthermore, the equality holds if and only if the resulting graph is a path, implying that the original graph has a hamiltonian path. 
On the other hand, if the original graph has a hamiltonian path, removing all other $|E(G)| - (|V(G)| - 1)$ edges will turn the graph into a path. Hence, $G$ is hamiltonian if and only if there is a $\Pi$-deletion of size $|E(G)| - (|V(G)| - 1)$.

Finally, since \ham is hard even for planar bipartite graphs of maximum degree 3 and chordal bipartite graphs, the same argument shows that $\Pi$-deletion is hard even restricted to these classes.
\end{proof}

\begin{corollary}
The $\Pi$-deletion problem is NP-hard when 
$\Pi$ is the class of proper interval graphs, claw-free chordal,
power of paths, even if the input is a bipartite graph. 
\end{corollary}

To the best of our knowledge, these are the first proofs of NP-hardness for claw-free chordal graphs and
power of paths. 
The case of proper interval graphs is proven in~\cite{LewisYannakakis80}; our contribution to this specific result is by restricting the input class to bipartite graphs. 

Notice that one could include paths in the previous corollary since they fulfill the listed properties in Theorem~\ref{teo:structural_prop}. However, the problem for paths is exactly the hamiltonian path problem, which is already known to be NP-hard even when the input is restricted to bipartite graphs. 


Finally, an interesting remark is that, if a class has bounded treewidth, it follows from~\cite{saitoh2020fixedtreewidthefficient} that deletion to proper interval graphs is polynomial-time solvable. In particular, this implies that the problem is easy for trees, cacti, and outerplanar graphs, which is a nice contrast with the hardness for planar bipartite graphs of maximum degree $3$ and chordal bipartite graphs.

\subsection{Subclasses of forests}\label{sec:subclasses_trees}


Inspired by the fact that modification problems turn out to be tractable when restricting the target class to trees, we study the deletion problem from general graphs to certain graph classes that are somewhat similar to trees. 
The first target class we consider is constellations. 

\paragraph{Constellations}

We show \delcons is equivalent to \domiset in the following way.

\begin{theorem}\label{teo:constellationsequivlance}
A graph $G$ has a dominating set of size $k$ if and only if it is possible to turn $G$ into a constellation by deleting $|E(G)|-|V(G)|+k$ edges.
\end{theorem}

\begin{proof}
Let $G=(V,E)$ be a graph.
Let $D \subseteq V$ be a dominating set of $G$ such that $|D| = k$. Let us consider the subgraph $H$ of $G$ as follows. Let $V(H)= D \cup N(D) = V(G)$, and for each $v \not\in D$, choose exactly one edge incident to a vertex $u \in D$. Those are all the edges of $H$, and we denote this set by $F$. Notice that $|F|=|V(G)|-|D|=|V(G)|-k$. Hence, $H$ can be obtained from $G$ by deleting $|E(G)|-|V(G)|+k$ edges.

Conversely, let $H$ be a spanning constellation of $G$ obtained by the deletion of  $|E(G)|-|V(G)|+k$ edges, and let $D =\{ u \mid u \mbox{ is the center of a star in } H \}$. Notice that $|E(H)| = |V(G)|-k$ and, since it is acyclic, it contains exactly $k$ connected components. Hence $|D| = k$, as the connected components of $H$ are stars and $D$ has precisely one vertex from each star. Since $D$ is a dominating set of $H$, it is also a dominating set of $G$ and the result follows.
\end{proof}

For the general case, we get the following corollary.

\begin{corollary}\label{teo:constellations}
The problem \delcons is NP-complete for general graphs.
\end{corollary}

\begin{remark}
Notice that \delcons remains hard for any graph classes for which \domiset is hard. In particular, this problem remains hard for planar bipartite graphs~\cite{zvervich1995induced} and split graphs~\cite{bertossi1984dominating,corneil1984clustering}. On the other hand, it can be polynomially solved within cographs and outerplanar graphs~\cite{vatshelle2012new}.
\end{remark}



\paragraph{Caterpillars}




We continue studying the general deletion problem by considering another subclass of trees, namely caterpillars, as the target class.

\begin{theorem}\label{teo:caterpillars}
The problem \delcater is NP-complete for general graphs.
\end{theorem}

\begin{proof}
The problem is clearly in NP.
We prove the hardness of \delcater by reducing from \ham.

Let $G=(V,E)$ be an instance of \ham. We define an instance $H,k$ of \delcater as follows. Let $H=(W, F)$ be a graph obtained by subdividing each edge of $G$ precisely twice. In other words, for each edge $e=uv \in E$, add two vertices $u_e, v_e$ such that $uu_e, u_ev_e, v_ev \in F$. \suggestion{}{Finally, let $k=|E(H)|-|V(H)|+1$.}

Suppose that there is a hamiltonian path in $G$. Consider the following spanning caterpillar $C$ of $H$. The spine is the hamiltonian path of $G$, where instead of traversing the edge $e=uv \in E$, we traverse the path $u, u_e, v_e, v$ in $H$. Notice that the only vertices that are not visited by the spine in $H$ are those intermediate new vertices $u_e, v_e$ for some $e=uv \in E$ that is not in the hamiltonian path. For these pairs of unvisited vertices, add the edges $u u_e$ and $v_e v$ to $C$. The resulting subgraph is a spanning caterpillar of $H$. \suggestion{}{Since every caterpillar is a tree, we know that $|E(C)|=|V(C)|-1=|V(H)|-1$, implying that $C$ can be obtained from $H$ by removing exactly $k=|E(H)|-|V(H)|+1$ edges.}

Conversely, \suggestion{a spanning caterpillar $C$ of $H$ has a spine $S$.}{let $C$ be a spanning caterpillar of $H$, with spine $S$. Clearly, $|E(H)-E(C)|=k$.} Let us see that we can extract a hamiltonian path in $G$ from this spine. 
Suppose that there is a vertex $u \in G$ that is not a part of the spine $S$ in $H$. Notice that, in a spanning caterpillar of $H$, all vertices should be incident to at least one edge in $C$. The vertex $u$ in $H$ is adjacent only to intermediate new vertices, thus there is an edge $u_eu \in C$ and such that $u_e \in S$. However, the only vertex adjacent to $u_e$ besides $u$ is another intermediate vertex $v_e$. Hence, both $u_e,v_e \in S$ and therefore we can consider $u$ as a vertex in $S$.
The hamiltonian path is precisely the path resulting of traversing the edges $uv \in G$ for which the path $u,u_e,v_e, v$ is a part of the spine of $C$ in $H$.
\end{proof}

\begin{remark}
Similarly to the case of \delcons, notice that the same proof given for Theorem~\ref{teo:caterpillars} suffices to show that \delcater remains hard for graphs classes \suggestion{}{that are closed for double subdivision of edges and} for which \ham is hard \suggestion{for double subdivision of edges}{}, which is the case for planar bipartite graphs~\cite{itai1982hamilton}.  
\end{remark}

\section{Algorithms for cactus deletion}\label{sec:algorithms}

For the natural classes considered in Section~\ref{sec:edge-del-complexity}, the complexity of \delcons and \delcater is very well understood, since they are closely related to the extensively studied \ham and \domiset problems. In this section we devote our attention to \delcactus. 
Recall that this problem is NP-complete even for bipartite graphs. 
This motivated the search for positive results on this problem.
In Section~\ref{sec:cactus_chordal} we study the problem within chordal graphs, showing that it turns out to be polynomial-time solvable. However, this result does not lead trivially to such an algorithm. With this in mind, in Section~\ref{sec:qt_cactus} we focus on subclasses of chordal graphs, and give an explicit efficient algorithm when the input is a quasi-threshold graph.


\subsection{\delcactus within chordal graphs}\label{sec:cactus_chordal}

In this section we will prove that cactus deletion from chordal graphs is polynomial-time solvable. To do this, we will resort to a result by Lovász~\cite{Lovasz80} which shows that, given a $3$-uniform hypergraph, finding the maximum Berge-acyclic subhypergraph can be solved in polynomial time.

We start this section by giving some definitions that will be useful in the sequel. Afterward, given a chordal graph $G$, we define an auxiliary hypergraph and use the aforementioned result by Lovász to find a maximum spanning cactus of $G$.

\begin{definition}
A \emph{hypergraph} is a pair $\mathcal{H} = (V,\mathcal{E})$ where $V$ is a finite set and $\mathcal{E}$ is a set of nonempty subsets of $V$. A \emph{subhypergraph} $\mathcal{S} = (V',\mathcal{E}')$ of $\mathcal{H}$ is a hypergraph such that $V' \subseteq V$ and $\mathcal{E}'\subseteq \mathcal{E}$.
Two vertices of a hypergraph are called \emph{adjacent} if some hyperedge contains both of them.

A \emph{Berge-cycle} is a sequence $(E_1, x_1, \ldots, E_n, x_n)$ with $n\geq 2$ such that:
\begin{enumerate}
    \item $E_i$ are distinct hyperedges,
    \item $x_i$ are distinct vertices, and
    \item for every $1 \leq i \leq n$, $x_i$ belongs to $E_i$ and $E_{i+1}$, where subindices are modulo $n$.
\end{enumerate}
A hypergraph is \emph{Berge-acyclic} if it contains no Berge-cycle.
\end{definition}

Let $G=(V,E)$ be a chordal graph. We denote by $\mathcal{H}(G) = (V, \mathcal{E})$ the hypergraph whose hyperedges are precisely those subsets that induce a $3$-cycle in $G$.


Notice that, if $G'$ is a subgraph of $G$ such that $G'$ is a cactus and each cycle has length $3$, then $\mathcal{H}(G')$ is a Berge-acyclic subhypergraph of $\mathcal{H}(G)$.
Also notice that two vertices are adjacent in 
$\mathcal{H}(G)$ if there is an hyperedge that contains both of them, this is, if both vertices lie in some $3$-cycle of $G$. In particular, this implies that there is an edge in $G$ that joins them.

Given a set of edges $E' \subseteq E$, we denote by $G[E']$ the subgraph $G'=(V', E')$ of $G$ where $V'=\{ v \in V(G) : \exists e \in E' \mbox{ incident to } v \}$.
Given a hyperedge $e \in \mathcal{E}$, we denote by $G[e]$ the subgraph of $G$ induced by edges of $G$ that are contained in the hyperedge $e$. Analogously, given a subset of hyperedges $\mathcal{E'} \subseteq \mathcal{E}$, we define $G[\mathcal{E'}]$ as $\bigcup_{e \in \mathcal{E'}} G[e]$. Intuitively, the vertices of $G[\mathcal{E'}]$ are the vertices that occur in hyperedges of $\mathcal{E'}$, and two vertices $u$ and $v$ are adjacent in $G[\mathcal{E'}]$ if there is some $e \in \mathcal{E'}$ such that $\{u,v\} \subset e$.


\begin{proposition}
Let $G$ be a chordal graph, and let $\mathcal{H'}(G)=(V', \mathcal{E'})$ be a Berge-acyclic subhypergraph of $\mathcal{H}(G)$. If $F_1, F_2, \ldots, F_k$ are the connected components of $G[\mathcal{E'}]$, then $F_1, F_2, \ldots, F_k$ is a forest of cacti and every cycle of $F_i$ is a $3$-cycle, for each $1 \leq i \leq k$.
\end{proposition}

\begin{proof}
First, let us see that every cycle in $F_i$ is a $3$-cycle. 
Toward a contradiction, let $C= x_1 x_2 \ldots x_n x_1$ be an induced cycle in $F_i$ of length $n>3$.
Since $C$ is an induced cycle, there are no chords between vertices in $C$. Thus, every hyperedge that intersects $C$ contains at most two vertices of $C$. Furthermore, each hyperedge that shares two vertices with $C$ contains precisely two consecutive vertices of $C$, for if not we find a chord in the cycle $C$. In addition, notice that no edge in $F_i$ lies in two distinct hyperedges of $\mathcal{H'}(G)$ since it is a Berge-acyclic subhypergraph of $\mathcal{H}(G)$. For simplicity, let us denote by $(E_i,x_i)$ \suggestion{to}{} the hyperedge in $\mathcal{H'}(G)$ that gives us the edge $x_i x_{i+1}$ in the cycle $C$. However, $\{(E_i, x_i) : i=1, \ldots, n \}$ is a Berge-cycle in $\mathcal{H'}(G)$ and this results in a contradiction.

Also notice that, if $C_1$ and $C_2$ are two distinct induced cycles in $F_i$, then $C_1$ and $C_2$ share at most one vertex. This follows once more from the fact that no edge in $F_i$ lies in two distinct hyperedges of $\mathcal{H'}(G)$. 
It follows immediately from the previous discussion that each $F_i$ is a cactus.
\end{proof}

Let $D \subseteq E$ be those edges that do not lie in any $3$-cycle of $G$. Notice that these edges are not considered when obtaining $\mathcal{H}(G)$.

\begin{remark}
Let $\mathcal{H'}(G)=(V', \mathcal{E'})$ be a maximum Berge-acyclic subhypergraph of $\mathcal{H}(G)$. Then, $G[\mathcal{E'}] \cup D$ is a forest of cacti.
\end{remark}

Consider first $G'= G[\mathcal{E'}] \cup D$, for some maximum Berge-acyclic hypergraph $\mathcal{H'}(G)=(V', \mathcal{E'})$ of $\mathcal{H}(G)$.
Let $F_1$ and $F_2$ be two connected components of $G'$ such that $F_1$ and $F_2$ are connected in the original graph $G$ by some missing edges of $G$. Since the missing edges belong to $3$-cycles that are not considered in $\mathcal{H'}(G)$, there exists a vertex $x$ in $F_1$ such that $x$ is adjacent (in $G$) to two distinct vertices $y_1, y_2$ in $F_2$, and such that $\{x, y_1, y_2\}$ is a hyperedge of $\mathcal{H}(G) \setminus \mathcal{H'}(G)$, or vice versa. Let us assume without loss of generality that $x$ \suggestion{}{is} in $F_1$. This gives a procedure to construct the spanning cactus, in the following way: 
\begin{enumerate}
\item Do while there is more than one connected component in $G'$
\item Find two connected components $F_1$ and $F_2$ such that they can be joined by a single edge $e$ that belongs to $E \setminus E(G')$. Define $F := F_1 \cup F_2 \cup \{e\}$.
\end{enumerate}
Let us denote with $X(G')$ the resulting subgraph. 

\begin{lemma}
Let $G$ be a \suggestion{}{connected} chordal graph. There is always an optimal spanning cactus of $G$ such that each cycle has length $3$. We call this kind of solution a triangular-SC of $G$. 
\end{lemma}

\begin{proof}
Suppose that the statement is false and let $H$ be a counterexample maximizing the number of cycles of length $3$ and, subject to that, minimizing the length of the minimum cycle of length greater than $3$. Let $C = \{x_1, \ldots, x_k\}$ be a minimum cycle with $ k > 3$. Since $G$ is chordal, $G[C]$ must contain a chord. Without loss of generality, let $x_1x_i$ be such chord. We claim that $H' = (V(H), E(H) - \{x_1x_2\} \cup \{x_1x_i\})$ is an optimal spanning cactus of $G$.

Note that $C' = \{x_1, x_i, \ldots, x_k\}$ is a cycle in $H'$, and since each edge of $E(C') \setminus \{x_1x_i\}$ belongs to exactly one cycle of $H$, it also belongs to exactly one cycle in $H'$, namely $C'$. Moreover, edges in $\{x_jx_{j+1} \mid 2 \leq j \leq i-1 \}$ belong to one cycle in $H$ and to no cycle in $H'$. Hence $H'$ is an spanning cactus of $G$ with at least the same number of cycles of length $3$ as $H$ and since the length of $C'$ is smaller than the length of $C$ the result follows.
\end{proof}

\noindent We have the following remark as a direct consequence of the previous lemma.

\begin{remark}\label{remark:chordal1}
A sub-optimal triangular-SC of $G$ always has fewer triangles than an optimal triangular-SC of $G$.
\end{remark}

Using all the previous results, it follows that:

\begin{theorem}
Let $G'= (V(G), E(G[\mathcal{E'}]) \cup D)$, for some maximum Berge-acyclic hypergraph $\mathcal{H'}(G)=(V', \mathcal{E'})$ of $\mathcal{H}(G)$. Then, $X(G')$ is a maximum spanning cactus of $G$. 
\end{theorem}
 
\begin{proof}
If this is not true, then there is an optimal triangular-SC $Y$ such that $|Y|>|X|$. It follows from Remark~\ref{remark:chordal1} that $Y$ has more triangles than $X$. If we consider the corresponding hypergraph for $Y$, \suggestion{}{with one hyperedge for each triangle in $Y$}, then this gives us a Berge-acyclic subhypergraph of $\mathcal{H}(G)$ that is bigger than the maximum $\mathcal{H'}(G)=(V', \mathcal{E'})$ that we found before. This results in a contradiction.
\end{proof}



\subsection{Cactus deletion from Quasi-Threshold}\label{sec:qt_cactus}

In the previous subsection, we provided proof that cactus deletion is polynomial when the input graph is chordal. Nevertheless, it is not trivial to derive an explicit algorithm from this result, ultimately based on matroid theory. The following results complete the section by giving a simple and efficient procedure for an interesting subclass of chordal graphs, namely, quasi-threshold graphs.

\begin{lemma}\label{lemma:join}
Let $G=G_1\oplus G_2$ be a cograph. Then, there is a spanning cactus subgraph $H$ with maximum number of edges such that each cycle of $H$ is either a $C_3$ or a $C_4$ and has a vertex in $V(G_1)$ and a vertex in $V(G_2)$.
\end{lemma}

\begin{proof}
Let $H$ be a spanning cactus subgraph of $G$ with the maximum number of edges, such that it contains a cycle $C_k = v_0, \ldots, v_{k-1}$.
If $H$ has a $C_k$ with $k > 4$ containing vertices from both $G_1$ and $G_2$, then $C_k$ has at least one chord $uv$ \suggestion{}{in $G$}, and we can replace $C_k$ with a smaller cycle by adding $uv$ and deleting an edge connecting $u$ and one of its neighbors in $C_k$. If $k \leq 4$ and $C_k$ has a vertex in $V(G_1)$ and a vertex in $V(G_2)$ the result follows.
Let us now assume without loss of generality that the vertices of $C_k$ are all contained in $V(G_1)$. 
We will construct another spanning cactus $H'$ from $H$ having the same number of edges, such that $C_k$ is transformed into a set of cycles, each of them having a vertex in $V(G_1)$ and a vertex in $V(G_2)$. The graph $H'$ starts as a copy of the spanning cactus $H$.
Let $v \in V(G_2)$, and let $v_i$ be the vertex of $C_k$ at a shortest distance from $v$ in $H'$. Notice that every path from a vertex of $C_k$ to $v$ in $H'$ passes through $v_i$. 
Consider now vertices $v_{i - 1}$ and $v_{i + 1}$ (the sums and subtractions are considered modulo $k$), which are the neighbours of $v_i$ in $C_k$. 
Delete edges $v_{i - 1}v_i$ and $v_i v_{i + 1}$ from $H'$. For the resulting path $P_{k - 1}$ obtained by removing these two edges from $C_k$, extract a maximum matching $M$ of $P_{k - 1}$ and delete the remaining edges $E(P_{k - 1}) \setminus M$ from $H'$. 
Finally, for each edge $e = u_1u_2 \in M$, add the edges $v u_1$ and $v u_2$ from $G$ to $H'$. Moreover, if $k-1$ is odd, then there is an unsaturated vertex $w$ in $P_{k-1}$. Add the edge $vw$ from $G$ to $H'$.
These edges exist in $G$, since $u_1, u_2, w \in V(G_1)$ and $v \in V(G_2)$. As the reader may easily verify, the resulting graph $H'$ is a spanning cactus of $G$, in which the cycle $C_k$ has been replaced by a collection of cycles having a vertex in $V(G_1)$ and a vertex in $V(G_2)$. Moreover, the new cactus $H'$ has $(k-1) - |E(P_{k - 1}) \setminus M| - 2$ more edges than $H$, which was optimal. We conclude that $k$ has to be either $3$ or $4$. Furthermore, the repeated application of this procedure allows to transform any spanning cactus $H$ having $C_3$ or $C_4$ entirely contained in $V(G_1)$ or $V(G_2)$ into another spanning cactus $H'$ in which these cycles have a vertex in $V(G_1)$ and a vertex in $V(G_2)$, and such that $|E(H')| = |E(H)|$. 
\end{proof}

\begin{theorem}
Let $G$ be a connected quasi-threshold graph. Then, a minimum deletion to cacti can be computed in polynomial time.
\end{theorem}

\begin{proof}
It follows from the definition of quasi-threshold graph that $G$ has a universal vertex $v$, since it is connected. In other words, $G$ is the join of a graph $G_1$ containing a single vertex and another quasi-threshold graph $G_2$. From Lemma~\ref{lemma:join}, it follows that there is a subgraph $H$ having the maximum possible number of edges such that $v$ belongs to each cycle of $H$. Without loss of generality, we may assume that each cycle is a triangle, otherwise we can transform a $C_4$ into a triangle plus an extra edge incident to $v$ using the same number of edges of the cactus.\\
Therefore, to find such cycles it is enough to find a maximum matching in $G_2$.
\end{proof}

\section{Conclusions and future work}\label{sec:conclusions}

In this work, we study the $\Pi$-deletion problem in some classes similar 
to trees. We obtain hardness results for cactus deletion restricted to \suggestion{subcubic}{} bipartite graphs, and for $\Pi$-deletion when $\Pi$ is 
is the property of belonging to the class of constellations or caterpillars. We also provide a list of properties that are sufficient to determine that, if the graph class $\Pi$ fulfills these conditions, then $\Pi$-deletion is NP-hard even when the input is a bipartite graph.
Given the hardness of these results, we then focus our study on positive results for \delcactus by restricting the input graph to chordal and quasi-threshold graphs. 
One possible future direction for this work is to deepen the study of \delcactus within other graph classes. We know that \delcactus is polynomial-time solvable for quasi-threshold graphs, thus a natural question arises regarding the complexity of this problem in superclasses of quasi-threshold graphs that are not chordal, such as cographs.
Another interesting graph class to restrict the input of the problem are planar graphs. 


\suggestion{}{We now briefly discuss Fixed Parameter Tractability (FPT) of the problem when considering the natural parameter $k$ (that is, the number of edges $k$ that we are allowed to delete). It is known that cacti can only have a number of edges linear on its number of vertices. Let $f(n)$ be the maximum number of edges in a cactus with $n$ vertices. A graph having more than $f(n) + k$ edges cannot be modified into a cactus with $k$ edge deletions. Thus, FPT of the problem becomes trivial if the graph has ``too many'' edges.
It is easy to see that, if $G$ is a graph and $H = (V(G), E(G) \setminus S)$ is a cactus, then the treewidth of $G$ is at most $2|S| + 2$. 
As a consequence of this, the clique-width of $G$ is also bounded by a function of $k$. 
This result also follows from~\cite{Bodlaender98} and the fact that cacti are precisely those graphs that do not contain the diamond as a minor.
Since minor-closed classes of graphs are MSO expressible~\cite{courcelle_engelfriet_2012}, it follows that there exists a linear-time FPT algorithm for \delcactus parameterized by $k$, the number of allowed deletions.
However, those algorithms would be impractical. Other research direction includes structural parameterizations, in particular those that do not necessarily grow with the density of the graph. Another direction for future work could be to restrict the input of cactus deletion to graph classes having bounded clique-width, in order to use their structural properties to design efficient algorithms for this problem.}


Finally, we present an interesting related question:
\textit{given a spanning tree $T$ of a graph $G$, is the problem of completing $T$ to a maximum spanning cactus by using only edges of $G$ polynomial?}
Notice that the same problem is NP-hard if we ask for a maximum spanning chordal graph: consider a graph $G$, and add a universal vertex $u$. Consider the spanning tree $T$ to be the star centered in the universal vertex $u$. Hence, the problem is equivalent to that of deletion to chordal in the original graph.

\section*{Acknowledgements}

\noindent Nina Pardal partially supported by a CONICET postdoctoral scholarship and UBACyT Grant \\ 20020190200124BA. Vinicius dos Santos was partially supported by CNPq Grants 312069/2021-9 and 406036/2021-7 and FAPEMIG Grant APQ-01707-21. 


 \bibliographystyle{abbrv} 
 \bibliography{cas-refs}

\begin{thebibliography}{10}

\bibitem{AlonSS05}
N.~Alon, A.~Shapira, and B.~Sudakov.
\newblock Additive approximation for edge-deletion problems.
\newblock In {\em 46th Annual {IEEE} Symposium on Foundations of Computer
  Science {(FOCS} 2005), 23-25 October 2005, Pittsburgh, PA, USA, Proceedings},
  pages 419--428. {IEEE} Computer Society, 2005.

\bibitem{bertossi1984dominating}
A.~A. Bertossi.
\newblock Dominating sets for split and bipartite graphs.
\newblock {\em Inf. Proc. Letters}, 19(1):37--40, 1984.

\bibitem{bertossi1986hamiltonian}
A.~A. Bertossi and M.~A. Bonuccelli.
\newblock Hamiltonian circuits in interval graph generalizations.
\newblock {\em Inf. Proc. Letters}, 23(4):195--200, 1986.

\bibitem{Bodlaender98}
H.~L. Bodlaender.
\newblock A partial \emph{k}-arboretum of graphs with bounded treewidth.
\newblock {\em Theor. Comput. Sci.}, 209(1-2):1--45, 1998.

\bibitem{Bondy}
J.~A. Bondy and U.~S. Murty.
\newblock {\em {G}raph {T}heory}.
\newblock Springer Publishing Company, Incorporated, 1st edition, 2008.

\bibitem{ChuMum1994}
F.~R. Chung and D.~Mumford.
\newblock Chordal completions of planar graphs.
\newblock {\em J. Comb. Theory Ser. B}, 62(1):96–106, sep 1994.

\bibitem{corneil1984clustering}
D.~G. Corneil and Y.~Perl.
\newblock Clustering and domination in perfect graphs.
\newblock {\em Disc. Appl. Math.}, 9(1):27--39, 1984.

\bibitem{courcelle_engelfriet_2012}
B.~Courcelle and J.~Engelfriet.
\newblock {\em Graph Structure and Monadic Second-Order Logic: A
  Language-Theoretic Approach}.
\newblock Encyclopedia of Mathematics and its Applications. Cambridge
  University Press, 2012.

\bibitem{Colburn-ElMallah88}
E.~El-Mallah and C.~Colbourn.
\newblock The complexity of some edge deletion problems.
\newblock {\em IEEE Transactions on Circuits and Systems}, 35(3):354--362,
  1988.

\bibitem{GolGolKapSha1995}
P.~W. Goldberg, M.~C. Golumbic, H.~Kaplan, and R.~Shamir.
\newblock Four strikes against physical mapping of dna.
\newblock {\em J. Comp. Biology}, 2:139--152, 1995.

\bibitem{GolKapSha1994}
M.~C. Golumbic, H.~Kaplan, and R.~Shamir.
\newblock On the complexity of dna physical mapping.
\newblock {\em Adv. Appl. Math.}, 15(3):251–261, sep 1994.

\bibitem{itai1982hamilton}
A.~Itai, C.~H. Papadimitriou, and J.~L. Szwarcfiter.
\newblock Hamilton paths in grid graphs.
\newblock {\em SIAM J. on Computing}, 11(4):676--686, 1982.

\bibitem{LewisYannakakis80}
J.~M. Lewis and M.~Yannakakis.
\newblock The node-deletion problem for hereditary properties is np-complete.
\newblock {\em J. Comput. Syst. Sci.}, 20(2):219--230, 1980.

\bibitem{Lovasz80}
L.~Lov{\'{a}}sz.
\newblock Matroid matching and some applications.
\newblock {\em J. Comb. Theory, Ser. {B}}, 28(2):208--236, 1980.

\bibitem{monnot2007path}
J.~Monnot and S.~Toulouse.
\newblock The path partition problem and related problems in bipartite graphs.
\newblock {\em Op. Res. Letters}, 35(5):677--684, 2007.

\bibitem{saitoh2020fixedtreewidthefficient}
T.~Saitoh, R.~Yoshinaka, and H.~L. Bodlaender.
\newblock Fixed-treewidth-efficient algorithms for edge-deletion to
  intersection graph classes.
\newblock In {\em {WALCOM:} Algorithms and Computation - 15th International
  Conference and Workshops, {WALCOM} 2021, Yangon, Myanmar, February 28 - March
  2, 2021, Proceedings}, pages 142--153, 2021.

\bibitem{TarYan1984}
R.~E. Tarjan and M.~Yannakakis.
\newblock Simple linear-time algorithms to test chordality of graphs, test
  acyclicity of hypergraphs, and selectively reduce acyclic hypergraphs.
\newblock {\em SIAM J. on Computing}, 13(3):566--579, 1984.

\bibitem{vatshelle2012new}
M.~Vatshelle.
\newblock {\em New width parameters of graphs}.
\newblock PhD thesis, The University of Bergen, 2012.

\bibitem{Yannakakis79}
M.~Yannakakis.
\newblock The effect of a connectivity requirement on the complexity of maximum
  subgraph problems.
\newblock {\em J. {ACM}}, 26(4):618--630, 1979.

\bibitem{yannakakis_edge_deletion}
M.~Yannakakis.
\newblock Edge-deletion problems.
\newblock {\em SIAM Journal on Computing}, 10(2):297--309, 1981.

\bibitem{zvervich1995induced}
I.~E. Zvervich and V.~E. Zverovich.
\newblock An induced subgraph characterization of domination perfect graphs.
\newblock {\em J. Graph Theory}, 20(3):375--395, 1995.

\end{thebibliography}





\end{document}